\documentclass[draft,DIV12,11pt,a4paper]{scrartcl}

\setkomafont{title}{\normalfont \LARGE \boldmath}
\setkomafont{section}{\normalfont \Large \bfseries \boldmath \boldmath}
\setkomafont{subsection}{\normalfont \large \bfseries \boldmath}
\setkomafont{subsubsection}{\normalfont \bfseries \boldmath}
\setkomafont{paragraph}{\normalfont \bfseries \boldmath}

\usepackage{amsmath,amssymb,amsfonts,amsthm}
\usepackage[boxed,linesnumbered]{algorithm2e}
\usepackage{cite,booktabs}
\usepackage{tikz}
\usepackage{authblk}

\usetikzlibrary{shapes}
\usetikzlibrary{calc}

\DontPrintSemicolon

\newcommand{\NP}{\ensuremath{\mathsf{NP}}}
\newcommand{\DP}{\ensuremath{\mathsf{P}}}
\newcommand{\APX}{\ensuremath{\mathsf{APX}}}

\makeatletter
\newtheoremstyle{newbold}{}{}{\itshape}{}{\bfseries}
   {.}{ }{\thmname{#1}\thmnumber{\@ifnotempty{#1}{ }\@upn{#2}}\thmnote{ {(#3)}}}
\makeatother

\theoremstyle{newbold}

\newtheorem{claim}{Claim}[section]
\newtheorem{theorem}[claim]{Theorem}
\newtheorem{lemma}[claim]{Lemma}
\newtheorem{corollary}[claim]{Corollary}
\newtheorem{remark}[claim]{Remark}

\newcommand{\RCS}[1]{\ensuremath{\textup{\textsf{$#1$-RCS}}}}
\newcommand{\RTwoCS}[1]{\ensuremath{\textup{\textsf{$#1$-R2CS}}}}
\newcommand{\ARCS}[1]{\ensuremath{\textup{\textsf{$#1$-ARCS}}}}
\newcommand{\MinRCS}[1]{\ensuremath{\textup{\textsf{Min-$#1$-RCS}}}}
\newcommand{\MinRTwoCS}[1]{\ensuremath{\textup{\textsf{Min-$#1$-R2CS}}}}
\newcommand{\MinARCS}[1]{\ensuremath{\textup{\textsf{Min-$#1$-ARCS}}}}
\newcommand{\MaxRCS}[1]{\ensuremath{\textup{\textsf{Max-$#1$-RCS}}}}
\newcommand{\MaxARCS}[1]{\ensuremath{\textup{\textsf{Max-$#1$-ARCS}}}}
\newcommand{\MST}{\ensuremath{\textup{\textsf{MST}}}}
\newcommand{\TSP}{\ensuremath{\textup{\textsf{TSP}}}}
\newcommand{\MinTSP}{\ensuremath{\textup{\textsf{Min-TSP}}}}
\newcommand{\MaxTSP}{\ensuremath{\textup{\textsf{Max-TSP}}}}
\newcommand{\MinATSP}{\ensuremath{\textup{\textsf{Min-ATSP}}}}
\newcommand{\MaxATSP}{\ensuremath{\textup{\textsf{Max-ATSP}}}}

\newcommand{\ATSP}{\ensuremath{\textup{\textsf{ATSP}}}}
\newcommand{\Fact}[1]{\ensuremath{\textup{\textsf{$#1$-F}}}}
\newcommand{\AFact}[1]{\ensuremath{\textup{\textsf{$#1$-AF}}}}

\newcommand{\eps}{\varepsilon}
\newcommand{\cut}{\ensuremath{\text{cut}}}

\title{Approximability of Connected Factors}
\author[1]{Kamiel Cornelissen}
\author[1]{Ruben Hoeksma}
\author[1]{Bodo Manthey}
\author[2]{N.~S.~Narayanaswamy}
\author[2]{C.~S.~Rahul}

\affil[1]{University of Twente, Enschede, The Netherlands \authorcr
\texttt{$\{$k.cornelissen, r.p.hoeksma, b.manthey$\}$@utwente.nl}}
\affil[2]{Indian Institute of Technology Madras, Chennai, India \authorcr
\texttt{$\{$swamy, rahulcs$\}$@cse.iitm.ac.in}}

\date{}

\begin{document}
\maketitle

\thispagestyle{plain}
\pagestyle{plain}

\begin{abstract}
Finding a $d$-regular spanning subgraph (or $d$-factor) of a graph is easy by Tutte's reduction to the matching problem. By the same reduction, it is easy to find a minimal or maximal $d$-factor of a graph. However, if we require that the $d$-factor is connected, these problems become NP-hard -- finding a minimal connected $2$-factor is just the traveling salesman problem (TSP).

Given a complete graph with edge weights that satisfy the triangle inequality, we consider the problem of finding a minimal connected $d$-factor. We give a 3-approximation for all $d$ and improve this to an $(r+1)$-approximation for even $d$, where $r$ is the approximation
ratio of the TSP. This yields a 2.5-approximation for even $d$. The same algorithm yields an $(r+1)$-approximation for the directed version of the problem, where $r$ is the approximation ratio of the asymmetric TSP. We also show that none of these minimization problems can be approximated better than the corresponding TSP.

Finally, for the decision problem of deciding whether a given graph contains a connected $d$-factor, we extend known hardness results.
\end{abstract}

\section{Introduction}
\label{sec:intro}

The traveling salesman problem (\MinTSP) is one of the basic combinatorial optimization problems:
given a complete graph $G=(V,E)$ with edge weights that satisfy the triangle inequality,
the goal is to find a Hamiltonian cycle of minimum total weight.
Phrased differently, we are looking for a subgraph of $G$ of minimum weight
that is 2-regular, connected, and spanning.
While \MinTSP\ is NP-hard~\cite[ND22]{GJ79}, omitting the requirement that the subgraph must be connected
makes the problem polynomial-time solvable~\cite{LPMatching,WT54}.
In general, $d$-regular, spanning subgraphs (also called $d$-factors) of minimum weight can be found in polynomial
time using Tutte's reduction~\cite{LPMatching,WT54} to the matching problem.
Cheah and Corneil~\cite{CC90} have shown that deciding whether a given graph $G=(V,E)$ has a $d$-regular connected spanning subgraph is
\NP-complete for every $d \geq 2$, where $d=2$ is just the Hamiltonian cycle problem~\cite[GT37]{GJ79}.
Thus, finding a connected $d$-factor of minimum weight is also \NP-hard for all $d$.

While one might think at first glance that the problem cannot become easier for larger $d$,
finding (minimum-weight) connected $d$-factors is easy for $d\geq n/2$, where $n = |V|$,
as in this case any $d$-factor is already connected. This poses the question for which values of $d$
(as a function of $n$) the problem becomes tractable.

In this paper, we analyze the complexity and approximability of the problem of finding a $d$-factor of minimum weight.

\subsection{Problem Definitions and Preliminaries}
\label{ssec:def}

In the following, $n$ is always the number of vertices. To which graph $n$ refers will be clear from the context.

All problems defined below deal with undirected graphs, unless stated otherwise.
For any $d$, $\RCS d$ is the following decision problem: Given an arbitrary undirected graph $G$, does $G$ have a connected $d$-factor?
Here, $d$ can be a constant, but also a function of the number $n$ of vertices of the input graph $G$.
\RCS 2\ is just the Hamiltonian cycle problem.

Just as \MinTSP\ is the optimization variant of \RCS 2, we consider the optimization variant
of \RCS d, which we call \MinRCS d:
As an instance, we are given an undirected complete graph $G=(V,E)$ and non-negative edge weights $w$ that satisfy the triangle inequality, i.e.,
$w(\{x,z\}) \leq w(\{x,y\}) + w(\{y,z\})$ for every $x, y, z \in V$.
The goal of $\MinRCS d$ is to find a connected $d$-factor of $G$ of minimum weight.
\MinRCS 2\ is just \MinTSP.

A \emph{bridge edge} of a graph is an edge whose removal increases the number of components
of the graph. A graph $G$ is called \emph{2-edge connected} if $G$ is connected and does not contain bridge edges.
For even $d$, any connected $d$-factor is also 2-edge-connected, i.e.,
does not contain bridge edges.
This is not true for odd $d$.
If we require 2-edge-connectedness also for odd $d$,
we obtain the problem
\MinRTwoCS d, which is defined as \MinRCS d, but asks for a 2-edge-connected
$d$-factor. For consistency, $\MinRTwoCS d$ is also defined for even $d$,
although it is then exactly the same problem as $\MinRCS d$.

Finally, we also consider the asymmetric variant of the problem:
given a directed complete graph $G=(V,E)$, find a spanning connected
subgraph of $G$ that is $d$-regular. Here, $d$-regular means that
every vertex has indegree $d$ and outdegree $d$. We denote the corresponding
minimization problem by \MinARCS d. \MinARCS 1\ is just the asymmetric TSP (\MinATSP).

\MaxRCS d\ and \MaxARCS d\ are the maximization variants
of \MinRCS d\ and \MinARCS d, respectively. For \MaxRCS d\ and \MaxARCS d we do not require that the edge weights satisfy the triangle inequality.
In the same way as for the minimization variants, \MaxRCS 2\ is the maximum TSP (\MaxTSP) and \MaxARCS 1\ is the maximum ATSP
(\MaxATSP).

If the graph and its edge weights are clear from the context, we abuse notation by also denoting by $\RCS d$ a minimum-weight connected $d$-factor,
by $\RTwoCS d$ a minimum-weight
2-edge-connected $d$-factor, and by $\ARCS d$ a minimum-weight connected $d$-regular subgraph of a directed graph.

In the same way, let $\Fact d$ denote a minimum-weight $d$-factor (no connectedness
required) of a graph and let $\AFact d$ denote a minimum-weight $d$-factor
of a directed graph.
Let $\MST$ denote a minimum-weight spanning tree, and let $\TSP$ and $\ATSP$ denote
minimum-weight (asymmetric) TSP tours.
We have $\RCS 2 = \TSP$ and $\ARCS 1 = \ATSP$. Furthermore,
$\Fact 2$ is the undirected cycle cover problem and $\AFact 1$ is the directed
cycle cover problem.

We note that $d$-factors do not exist for all combinations of $d$ and $n$. If both $n$ and $d$ are odd,
then no $n$-vertex graph possesses a $d$-factor. For all other combinations of $n$ and $d$ with $d\leq n-1$, there exist $d$-factors
in $n$-vertex graphs, at least in the complete graph.

In the following, $K_n$ denotes the undirected complete graph on $n$ vertices.
A vertex $v$ of a graph $G$ is called a \emph{cut vertex} if removing $v$ increases the number of components of $G$.

\subsection{Previous Results}
\label{ssec:previous}

Requiring connectedness in addition to some other combinatorial property has already been studied for
dominating sets~\cite{GuhaKhuller} and vertex cover~\cite{EscoffierEA}. For problems such as minimum $s$-$t$ vertex separator, which are known
to be solvable in polynomial time, the connectedness condition makes it NP-hard, and recent results have studied the parameterized complexity of
finding a connected $s$-$t$ vertex separator~\cite{Marx2013}.
Also finding connected graphs with given degree sequences that are allowed to be violated only slightly has been
well-studied~\cite{BND,SinghLau:TreeOne:2007}.

As far as we are aware, so far only the maximization variant \MaxRCS d\
of the connected factor problem has been considered for $d \geq 3$. Baburin, Gimadi, and
Serdyukov proved
that \MaxRCS d\ can be approximated within a factor of
$1-\frac{2}{d \cdot (d+1)}$~\cite{BG:2004,GimadiS}. A
slightly better approximation ratio can be achieved if the edge
weights are required to satisfy the triangle inequality~\cite{BG:2006}.
Baburin and Gimadi also considered approximating both \MaxRCS d\ and \MinRCS d\
(both without triangle inequality) for random instances~\cite{BG:2006,BG:2008}.
For $d=2$, we inherit the approximation results for
\MinTSP\ of $3/2$~\cite[Section 2.4]{Shmoys2011} and \MaxTSP\ of $7/9$~\cite{MaxTSP79}. For $d=1$, we inherit
the $O(\log n/\log \log n)$-approximation for \MinATSP~\cite{ATSPolog} and
$2/3$ for \MaxATSP~\cite{KaplanEA2005}.
As far as we know, no further polynomial-time approximation algorithms
with worst-case guarantees are known for \MinRCS d.
Like for \MinTSP~\cite[Section 2.4]{Shmoys2011}, the triangle
inequality is crucial for approximating
\MinRCS d\ and \MinARCS d\ -- otherwise, no polynomial-time approximation algorithm
is possible, unless $\DP = \NP$.
Baburin and Gimadi~\cite{BG:2004,BG:2006} claimed that \MaxRCS d is \APX-hard because it generalizes \MaxTSP. However,
this is only true if we consider $d$ as part of the input, as then $d=2$ corresponds
to \MaxTSP.

\subsection{Our Results}
\label{ssec:results}

Table~\ref{tab:overview} shows an overview of previous results and our results.

Our main contributions are a 3-approximation algorithm for \MinRCS d\ for any $d$
and a $2.5$-approximation algorithm for \MinRCS d\ for even $d$ (Section~\ref{sec:approx}).
The latter is in fact an $(r+1)$-approximation algorithm for \MinRCS d, where $r$ is the factor within which 
\MinTSP\ can be approximated. This result can be extended to \MinARCS d, where $r$ is now the approximation ratio
of \MinATSP. Our approximation algorithms, in particular for the maximization variants, are in the spirit
of the classical approximation algorithm of Fisher et al.~\cite{FisherEA} for \MaxTSP: compute a non-connected structure, and then remove and add
edges to make it connected.

As lower bounds, we prove that \MinRCS d\ and \MinARCS d\ cannot be approximated
better than \MinTSP\ and \MinATSP, respectively (Section~\ref{sec:hardness}). In particular, this implies the \APX-hardness
of the problems.

We prove some structural properties of connected $d$-factors
and their relation to $\TSP$, $\MST$, and $d$-factors without connectedness requirement
(Section~\ref{sec:structural}). Some of these properties are
needed for the approximation algorithms and some might be interesting
in their own right or were initially counterintuitive to us.

Our algorithms work for all values of $d$, even when $d$ is part of the input.
The hardness results are extended to the case where $d$ grows with $n$.
In Section~\ref{sec:special}, we improve our approximation
guarantee for $d \geq n/3$, prove that $\RCS{(\frac n2 -1)} \in \DP$,
and generalize Baburin and Gimadi's algorithm~\cite{BG:2004} to directed instances.

\begin{table}[t]
\centering
    \begin{tabular}{lll} \toprule
   \emph{problem} & \emph{result} & \emph{reference} \\ \midrule
   \RCS d  & in \DP\ for $d \geq \frac n2 -1$ & trivial for $d \geq n/2$,
     Section~\ref{ssec:decision} \\ 
   & \NP-complete for constant $d$ & Cheah and Corneil~\cite{CC90} \\
    & and $d$ of any growth rate up to $O(n^{1-\eps})$ & Section~\ref{ssec:growingd} \\
    \midrule
   \MinRCS d & $(r +1)$-approximation for even $d$ & Section~\ref{ssec:tspapprox} \\
            & 3-approximation for odd $d$ & Section~\ref{ssec:threeapx} \\ 
            & 2-approximation for $d \geq n/3$ & Section~\ref{ssec:twoapprox} \\
						& no better approximable than \MinTSP\ & Section~\ref{ssec:tspinapprox} \\ \midrule
   \MinRTwoCS d & 3-approximation & Section~\ref{ssec:threeapx} \\ 
			& no better approximable than \MinTSP\ & Section~\ref{ssec:tspinapprox} \\ \midrule
   \MinARCS d & $(r +1)$-approximation & Section~\ref{ssec:tspapprox} \\
							& no better approximable than \MinATSP\ & Section~\ref{ssec:tspinapprox} \\ \midrule
   \MaxRCS d & $(1-\frac{2}{d \cdot (d+1)})$-approximation& Baburin and Gimadi~\cite{BG:2004}\\
\midrule
   \MaxARCS d & $(1 - \frac{1}{d \cdot (d+1)})$-approximation & Section~\ref{ssec:directedapx} \\
\bottomrule
  \end{tabular}
\caption{Overview of the complexity and approximability of finding (optimal) connected $d$-factors.
We left out that all optimization variants are polynomial-time solvable for $d \geq n/2$ and \APX-hard
according to Sections~\ref{ssec:tspinapprox} and~\ref{ssec:growingd}. Here, $r$ is the approximation
ratio of \MinTSP\ or \MinATSP.}
\label{tab:overview}
\end{table}

\section{Structural Properties}
\label{sec:structural}

In the following two lemmas, we make statements about the relationship between the
weights of optimal solutions of the different minimization problems.
We call an inequality $A \leq c \cdot B$ \emph{tight}
if, for every $\eps > 0$, replacing $c$ by $c-\eps$ does not yield a valid statement
for all instances.

\begin{lemma}[undirected comparison]
\label{lem:undirected}
\begin{enumerate}
\item $w(\MST) \leq w(\RCS d) \leq w(\RTwoCS d)$ for all $d$ and all undirected instances, and this is tight. \label{mstrcs}
\item $w(\Fact d) \leq w(\RCS d)$ for all $d$ and all undirected instances, and this is tight.
\label{factrcs}
\item $w(\RTwoCS d) \leq 3 \cdot w(\RCS d)$ for all odd $d$ and all undirected instances, and this is tight for all odd $d$. \label{onetwo}
\item $w(\TSP) \leq w(\RCS d)$ for all even $d$ and all undirected instances, and this is tight.
\label{tspeven}
\item $w(\TSP) \leq 2 \cdot w(\RCS d)$ for all odd $d$ and all undirected instances, and this is tight for all odd $d$. \label{tspodd}
\item $w(\TSP) \leq \frac 43 \cdot w(\RTwoCS 3)$ for all undirected instances, and this is tight.
\label{fourthird3}
\item For all odd $d$, there are instances with
$w(\TSP) \geq (\frac 43 - o(1)) \cdot w(\RTwoCS d)$. \label{fourthird}
\item $w(\Fact{(d-2)}) \leq \frac{d-2}d \cdot w(\Fact d)$ and $w(\RCS{(d-2)}) \leq w(\RCS{d})$ for all even $d \geq 4$ and all undirected instances, and both inequalities
are tight. \label{evenmonotone}
\item Monotonicity does not hold for odd $d$: for every odd $d \geq 5$, there exist instances with
  $w(\RCS{(d-2)}) \geq \frac{d+2}{d} \cdot w(\RCS d)$. \label{nonmonodd}
\end{enumerate}
\end{lemma}

\begin{proof}
\emph{Items~\ref{mstrcs} and~\ref{factrcs}:}
The inequalities follow immediately from the definitions.
The inequality of Item~\ref{factrcs} and the second inequality of Item~\ref{mstrcs} are tight for instances
where all edge weights are equal. That the first inequality of Item~\ref{mstrcs} is also tight can be seen as follows:
For any $k\geq 2$ we can construct a graph consisting of $k$ groups of $d+1$ vertices. We set the distance between each pair of vertices from the same group equal to 0, and the distance between each pair of vertices from different groups equal to $1$.
Any MST of the new instance has a weight of $k-1$. We can construct a (2-edge-)connected $d$-factor of weight $k$ by combining a global TSP tour
with a $d-2$-factor within each group. Since this holds for all $k$, the first inequality of Item~\ref{mstrcs} is tight.

\emph{Item~\ref{onetwo}:} The inequality is shown to hold constructively by Algorithm~\ref{threeapx},
which computes a 2-edge-connected $d$-factor that weighs no more than three times
the weight of a minimum-weight connected $d$-factor. It is tight because of the
following example: 
The set of vertices of the instance consists of a vertex $v$ plus $d$ sets $V_1, \ldots, V_d$ which consist of $d+2$ vertices each.
We set the distance of $v$ to each of the other vertices equal to $1$. The distance between each pair of vertices from the same set $V_i$ is $0$. Finally, the distance between all pairs of vertices from different sets $V_i$ and $V_j$ is $2$. An optimal connected $d$-factor connects $v$ to one vertex
of each $V_i$. Since each $V_i$ has an odd number of vertices and $d$ is odd as well, we can complete
the connected $d$-factor without any further cost. Thus, the total cost is $d$.

An optimal 2-edge-connected $d$-factor has a weight of $3d$: Because each $V_i$
has an odd number of vertices and $d$ is odd, any 2-edge-connected $d$-factor must have
at least three edges leaving each set $V_i$. If such an edge $e$ is incident with $v$, then we charge
its weight to $V_i$. The other possibility is that $e$ is incident with a vertex from some $V_j$, where $j \neq i$. In this case $e$
has a weight of $2$ and we charge a weight of $1$ to both $V_i$ and $V_j$. The total
charge of all sets $V_I$ equals the total weight of the 2-edge-connected $d$-factor. Since each $V_i$ is charged at least $3$, the total weight of any 2-edge-connected $d$-factor is at least $3d$.

\emph{Item~\ref{tspeven}:} Since $d$ is even, \RCS d\ is Eulerian and we can take shortcuts to obtain a TSP tour whose weight
is no larger than w(\RCS d). This is tight because $w(\TSP) \geq w(\MST)$ and
Item~\ref{mstrcs}.

\emph{Item~\ref{tspodd}:} If we duplicate each edge of a connected $d$-factor, we obtain a Eulerian
multi-graph. Taking shortcuts yields a TSP tour and proves the inequality. The same
instance as we used in the proof of Item~\ref{onetwo} shows that this is tight, as every set $V_i$ is charged at least $2$ for any TSP tour.

\emph{Item~\ref{fourthird3}:} We exploit the following property of 3-regular 2-edge-connected graphs:
for every odd subset $U$ of the vertices, there are at least three edges that connect $U$ to $V \setminus U$.
This implies that the fractional edge-coloring number of such a graph is 3~\cite[Corollary 28.5a]{Schrijver}.
We consider an optimal solution $\RTwoCS 3$. The above implies that there exists a collection of -- not necessarily disjoint --
matchings $M_1, \ldots, M_k$ for some $k$
as well as non-negative numbers $\lambda_1, \ldots, \lambda_k$ such that $\sum_{i: e \in M_i} \lambda_i = 1$
for all edges $e \in \RTwoCS 3$ and $\sum_{i=1}^k \lambda_i=3$ (Seymour~\cite{Seymour:Multicolorings:1979} attributes this to Edmonds~\cite{Edmonds:Polyhedron:1965}). This implies $\sum_{i=1}^k \lambda_i w(M_i) = w(\RTwoCS 3)$.
Since the total number of edges in $\RTwoCS 3$ equals $3n/2$ and the $\lambda_i$'s sum to 3, all the matchings  $M_1, \ldots, M_k$ are necessarily perfect. Also, since $\sum_{i=1}^k \lambda_i = 3$,
there exists an $i$ with $w(M_i) \leq \frac 13 \cdot w(\RTwoCS 3)$. Let $M$ be a minimum-weight perfect matching of the instance.
Then $w(M) \leq w(M_i)$. Furthermore, adding $M$ to $\RTwoCS 3$ yields a 4-regular connected multi-graph
of weight at most $\frac 43 \cdot w(\RTwoCS 3)$, which is thus Eulerian. Taking shortcuts yields a TSP tour, which proves the claim.
It is tight by Item~\ref{fourthird}.

\emph{Item~\ref{fourthird}:}
For arbitrary odd $d$ we construct a complete graph $G=(V,E)$ as follows:
$G$ consists of $3m-1$ gadgets $A$, $B$, and $U_{i,j}$ for $i \in \{1,2,3\}$ and $j \in \{1,\ldots, m-1\}$. The gadgets
$A$ and $B$ consist of $d+2$ vertices, all $U_{i,j}$ consist of $d+1$ vertices.
All pairs of vertices within the same
gadget have a distance of $0$. All edges between $A$ and $U_{i, 1}$, between
$U_{i, j}$ and $U_{i, j+1}$, and between $U_{i, m-1}$ and $B$ have a weight of one for all $i \in \{1,2,3\}$ and
$j \in \{1, \ldots, m-2\}$.
All other distances are obtained by metric completion, i.e., by taking the shortest path distances. Thus, e.g.,
the distance between a vertex in $A$ and a vertex in $B$ is $m$. Figure \ref{fig:upperbound} depicts the construction.

\begin{figure}[t]
\centering
\newcommand{\xstretch}{1.8}
\newcommand{\ystretch}{1.4}
\begin{tikzpicture}
   \tikzstyle{DotsNode}=[ellipse, line width=0.5pt, inner sep=0pt, minimum width=14mm, minimum height=5mm]
   \tikzstyle{Node}=[DotsNode, draw]
   \tikzstyle{Edge}=[line width=0.7pt]
	 \tikzstyle{Cut}=[Edge, dotted]
	 \coordinate (cutshift) at (-0.5*\xstretch, 0);
	 \node[Node] (A) at (0,-2*\ystretch) {$A$};
	 \node[Node] (B) at (5*\xstretch,-2*\ystretch) {$B$};
	 \foreach \i in {1,2,3}
	   {\foreach \j/\k in {1/1,2/2,4/m-1}
		    {\node[Node] (U\i\j) at (\j*\xstretch, -\i*\ystretch) {$U_{\i,\k}$};}
				 \node[DotsNode] (U\i3) at (3*\xstretch, -\i*\ystretch) {$\cdots$};
         \draw[Edge] (A) -- (U\i1);
				 \draw[Edge] (B) -- (U\i4);
				 \foreach \j/\k in {1/2, 2/3, 3/4}
				   {\draw[Edge] (U\i\j) -- (U\i\k);}}
   \foreach \j/\k in {1/1, 2/2, 5/m}
	   {\draw[Cut] ($(\j*\xstretch, -3.1*\ystretch)+(cutshift)$)
	            -- ($(\j*\xstretch, -0.9*\ystretch)+(cutshift)$) node[above] {$\cut_{\k}$};}
\end{tikzpicture}
\caption{Graph with $w(\TSP) \geq (\frac 43 - o(1)) \cdot w(\RTwoCS d)$. Gadgets $A$ and $B$ consist
			of $d+2$ vertices, all $U_{i,j}$ consist of $d+1$ vertices. Edges between a pair of vertices in the same gadget have weight 0. An edge between two gadgets indicates that the weight of an edge between a vertex in the first gadget and a vertex in the second gadget equals $1$. All other distances are obtained via metric completion.}
			\label{fig:upperbound}
\end{figure}
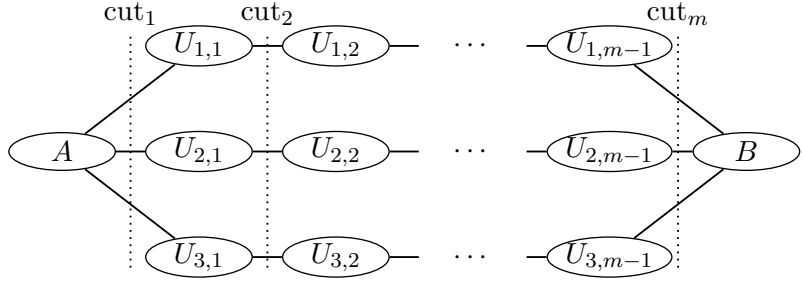

We build a $d$-factor of total weight $3m$ as follows: For each pair of gadgets that has a distance of $1$, we take an edge between the two gadgets and include it in the $d$-factor. We do so in a way that all selected edges of weight $1$ are disjoint. We complete the $d$-factor by taking appropriate edges of weight 0 within the gadgets.
By the choice of the size of the gadgets, this can be done. The $d$-factor obtained is even 2-edge-connected.

Now we show that $w(\TSP) \geq 4m-2$.
Let $T$ be a minimum-weight TSP tour on $G$. Consider $T'$, the multigraph where each edge in $T$ of length greater than $0$ is replaced by its shortest path over edges with weight $1$. Then $w(T')=w(T)$.
For $j \in \{2, \ldots, m-1\}$, let $\cut_j$ be the sum over all $i$ of the number of edges connecting $U_{i,j-1}$ to $U_{i, j}$.
Let $\cut_1$ be the number of edges connecting $A$ to $U_{1, 1}$, $U_{2,1}$, and $U_{3,1}$,
and let $\cut_m$ be the number of edges connecting $U_{1, m-1}$, $U_{2,m-1}$, and $U_{3,m-1}$ to $B$.
(The cuts are indicated by dotted lines in Figure~\ref{fig:upperbound}.)
Since $T'$ uses only edges within gadgets or edges of weight $1$, we have $w(T') = \sum_{j=1}^m \cut_j$.
Since $T'$ is Eulerian, we know that $\cut_j$ is even for all $j$. Since $T'$ is connected, $\cut_j \geq 2$ for all $j$.
If $\cut_j \geq 4$ for all $j$, then $w(T') \geq 4m-2$, and we are done.
Otherwise, $\cut_j = 2$ for some $j$. For ease of notation, we assume that $j \in \{2, \ldots, m-1\}$. The remaining cases
are almost identical.
Then there is some $U_{i, j-1}$ that is not connected to $U_{i, j}$ in $T'$. Thus, there must be two paths in $T'$
that connect $U_{i, j-1}$ via $A$ to $U_{i', j-1}$ and $U_{i'', j-1}$ ($i', i'' \neq i$, but $i'=i''$ is allowed).
In the same way, $T'$ must contain two paths from $U_{i, j}$ via $B$ to $U_{i', j}$ and $U_{i'',j}$.
The weight of the former paths is $2j-2$ each. The weight of the latter paths is $2m-2j$ each. In total,
the weight is $4m-4$. Adding $\cut_j$ yields the result.

\emph{Item~\ref{evenmonotone}:} The inequality $w(\Fact{(d-2)}) \leq \frac{d-2}d \cdot w(\Fact d)$ holds
since every $d$-factor for even $d$ can be split into $d/2$ 2-factors, and we can remove the lightest 2-factor to obtain a $(d-2)$-factor.
This is tight if all edge weights are equal.

Now consider the claim for connected factors. For $n=d+1$, the $d$-factor is the complete graph and the claim is trivial.
Thus, we assume $n>d+1$.
Let $R$ be a connected $d$-factor.
We describe a process to obtain a connected $(d-2)$-factor from $R$ that weighs at most the weight of $R$. To do this, we use the following invariants:
First, if a graph has only even degrees and is connected, then it is 2-edge-connected.
(Otherwise, removing a bridge edge would result in two components, and the sum of degrees within each component would be odd.)
Second, if the maximum degree of a graph is $d \leq n-2$ and the graph is connected, then any vertex with degree $d$ is adjacent
to at least two vertices $x$ and $y$ such that $x$ and $y$ are not adjacent.

Let us now describe the process:
We take any vertex $v$ with degree $d$. If $v$ is not a cut vertex, then we choose a pair of edges $\{v,x\}$ and $\{v,y\}$ of $R$ such that $\{x,y\} \notin R$. This exists due to the invariant. We replace $\{v,x\}$ and $\{v,y\}$ by $\{x,y\}$, which reduces the degree of $v$ by two and does not change the degree of the other vertices.
The graph is still connected, as $v$ has still a degree of at least $d - 2 \geq 2$ and was no cut vertex.

If $v$ is a cut vertex, let $C_1,  \ldots, C_k$ be the components obtained by removing $v$. Vertex $v$ has at least two edges to each component
by 2-edge-connectedness.
We take one edge to $x\in C_1$ and one edge to $y\in C_2$ and shortcut it, i.e, remove $\{v,x\}$ and $\{v,y\}$ and add
$\{x,y\}$. Again this reduces the degree of $v$ by two and does not change the degree of any other vertex. Also by 2-edge-connectedness,
it does not disconnect the graph.
Repeating this process until all vertices have degree $d-2$ yields a connected $(d-2)$-factor. By the triangle inequality,
its weight is no more than the weight of the connected $d$-factor.

The analysis is tight because of Item~\ref{tspeven}.

\emph{Item~\ref{nonmonodd}:} Consider the following instance: we have a central vertex $v$
and $d$ sets $V_1, \ldots, V_{d}$, which each consist of $d+2$ vertices. The distance of $v$ to each other vertex is $1$.
Within any set $V_i$, the distances are $0$. The distance between each pair of vertices in two different sets $V_i$ and $V_j$ is $2$.

The optimal connected $d$-factor has a weight of $d$: we connect $v$ to one vertex of each set $V_i$ and complete
the $d$-factor locally within each $V_i$. In any connected $(d-2)$-factor, $v$ is connected to $k \leq d-2$
of the sets $V_1, \ldots, V_{d}$. Thus, there are $d-k$ sets that have to be connected with an edge of weight $2$. This results in
costs of at least $k + 2 \cdot (d-k) \geq d+2$. 
\end{proof}

\begin{lemma}[directed comparison]
\label{lem:directed}
\begin{enumerate}
\item $w(\AFact d) \leq w(\ARCS d)$ for all $d$ and all directed instances, and this is tight.
\label{arcsfact}
\item $w(\ATSP) \leq w(\ARCS d)$ for all $d$ and all directed instances, and this is tight.\label{arcsatsp}
\item $w(\AFact{(d-1)}) \leq \frac{d-1}{d} \cdot w(\AFact d)$ and $w(\ARCS{(d-1)}) \leq w(\ARCS d)$ for all $d \geq 2$
and all directed instances, and both inequalities are tight. \label{dir:monotone}
\end{enumerate}
\end{lemma}

\begin{proof}
The inequality of Item~\ref{arcsfact} holds because we optimize over a larger set to obtain
$\AFact d$.
It is tight, e.g., if all edge weights are equal.

The inequality of Item~\ref{arcsatsp} holds since \ARCS d\ is Eulerian and connected.
Thus, we can obtain a TSP tour by taking shortcuts and the triangle inequality
guarantees that this does not increase the weight.
It is tight, e.g., for the following instance: We have two clusters of $d+1$ vertices each,
and within each cluster, the weights are $0$. Between the clusters, the weights
are $1$. Both $\ATSP$ and $\ARCS d$ use only two edges between the two clusters,
one in each direction.
Thus, they have equal weight.

The first part of Item~\ref{dir:monotone} is straightforward, because $d$-regular directed graphs
stand in one-to-one correspondence to $d$-regular bipartite graphs, whose edges can be partitioned into $d$
perfect matchings~\cite[Lemma~1.4.17]{LPMatching}. Tightness holds if all edge weights are equal.

The tightness of the second inequality of Item~\ref{dir:monotone}
follows from the tightness of Item~\ref{arcsatsp} and the fact that $\ARCS 1 = \ATSP$.
Its proof is similar to the proof of Lemma~\ref{lem:undirected}\eqref{evenmonotone}:
Let $R$ be a minimum-weight connected $d$-factor. We iteratively decrease indegree and outdegree of every vertex by one.
The invariant that we use is that the graph is strongly connected. This holds for every weakly connected directed graph
where indegree equals outdegree at every vertex.

Consider a vertex $v$ with indegree $d$ and outdegree $d$. First, assume that $v$ is a cut vertex, and let $C_1,\ldots, C_k$
be the components obtained by removing $v$. Since $R$ is strongly connected,
there is an $x \in C_1$ and $y \in C_2$ with $(x,v), (v, y) \in R$.
We remove these two edges and add $(x,y)$. The resulting graph is weakly connected and thus strongly connected.

If $v$ is not a cut vertex, let $y$ be arbitrary with $(v,y) \in R$. Since $y$ has indegree at most $d$,
there must be an $x$ with $(x,v) \in R$ and $(x,y) \notin R$. We replace $(x,v)$ and $(v,y)$ by $(x,y)$.
Since $v$ is no cut vertex, the resulting graph is connected.
We iterate this process until we have a $(d-1)$-regular connected graph. By the triangle inequality,
its weight is at most $w(\ARCS d)$. 
\end{proof}

\section{Approximation Algorithms}
\label{sec:approx}

\subsection{3-Approximation for \MinRCS d\ and \MinRTwoCS d}
\label{ssec:threeapx}

The 3-approximation that we present in this section works for all $d$, odd or even. It also works for $d$ growing as a function of $n$.
An interesting feature of this algorithm, and possibly an indication that a better approximation ratio is possible for \MinRCS d,
is that the same algorithm provides an approximation ratio of 3 for both
\MinRCS d\ and \MinRTwoCS d. In fact, we compute a 2-edge-connected $d$-regular graph
that weighs at most three times the weight of the optimal connected $d$-regular graph.

First we make some preparatory observations on 2-edge-connectedness.
Given a connected graph $G=(V,E)$, we can create a tree $T(G)$ as follows: We have a vertex for every maximal subgraph of $G$
that is 2-edge-connected (called a 2-edge-connected component), and two such vertices are connected if the corresponding components are connected
in $G$. In this case, they are connected by a bridge edge.
Now consider a leaf of tree $T(G)$ and its corresponding 2-edge-connected component $C$. Since
$C$ is a leaf in $T(G)$, it is only incident to a single bridge edge $e$ in $G$.
Now assume that $G$ is $d$-regular with $d\geq 3$ odd (for $d=2$, any connected graph is also 2-edge-connected).
Let $u$ be the vertex of $C$ that is incident to $e$. Then $u$ must be incident to $d-1$ other vertices in $C$.
Thus, $C$ has at least $d$ vertices. Since the $d-1$ neighbors of $u$ are not incident to bridge edges, they must be adjacent
to other vertices in $C$. Since $G$ is $d$-regular, $C$ has at least $d+1$ vertices and more than $d^2/2 > d$ edges.
Therefore, there exists an edge $e'$ in $C$ that is not incident to $u$, i.e., $e'$ does not share an endpoint with a bridge edge.

If $G$ is not connected, we have exactly the same properties with ``tree'' replaced by ``forest''.

To simplify notation in the algorithm, let $k=k(G)$ denote the number of 2-edge-connected
components of $G$ that are leaves in the forest described above,
and let $L_1(G), \ldots, L_{k}(G)$ denote the 2-edge-connected components of a graph $G$
that correspond to leaves in the tree described above. For such an $L_i(G)$, let $e_i(G)$ denote an edge that is not adjacent
to a bridge edge in $G$. The choice of $e_i(G)$ is arbitrary.

\begin{algorithm}[t]
\SetKwInOut{Input}{input}\SetKwInOut{Output}{output}
\Input{undirected complete graph $G=(V,E)$, edge weights $w$, $d \geq 2$}
\Output{2-edge-connected $d$-factor $R$ of $G$}

compute a minimum-weight $d$-factor $\Fact d$\ of $G$;

$k \leftarrow k(\Fact d)$\;

$Q \leftarrow \{e_1, \ldots, e_{k}\}$ with
$e_i = e_i(\Fact d) = \{u_i, v_i\}$\; \label{treeconstruction}

compute \MST\ of $G$;\label{dt1}

duplicate each edge of \MST\ and take shortcuts to obtain a Hamiltonian cycle $H$\;\label{dt2}

take shortcuts to obtain from $H$ a Hamiltonian cycle $H'$ through $\{u_1, \ldots, u_k$\}, assume w.l.o.g.\
that $H'$ traverses the vertices in the order $u_1, \ldots, u_k, u_1$\;

obtain $R$ from $\Fact d$ by adding the edges $\{u_i, v_{i+1}\}$ (with $k+1 = 1$) and removing~$Q$\;
\caption{3-approximation for \MinRCS d\ and \MinRTwoCS d.}
\label{threeapx}
\end{algorithm}

We prove that Algorithm~\ref{threeapx} is a 3-approximation for both \MinRCS d\ and \MinRTwoCS d\ by a series of lemmas.
Since the set of vertices is fixed, we sometimes identify graphs with their edge set. In particular,
$R$ denotes both the connected $d$-factor that we compute and its edge set.

\begin{lemma} \label{lem:spanning}
Assume that $R$ is computed as in Algorithm~\ref{threeapx}.
Then $R$ is a $d$-regular spanning subgraph of $G$.
\end{lemma}

\begin{proof}
If $R$ does not contain multiple edges between the same pair of vertices, then $R$ is $d$-regular since
we obtain $R$ from a $d$-factor, remove one edge incident to each $u_i$ and $v_i$, and add one edge incident to each $u_i$ and $v_i$.
We now show that indeed we have that $R$ does not contain multiple edges between the same pair of vertices. This can only happen if some edge $e = \{u_i, v_{i+1}\}$
is present in $\Fact d$. Since $e$ connects two 2-edge-connected components, this can only happen
if $e$ is a bridge edge. This is not the case as none of the vertices $u_i$ and $v_i$ are incident to a bridge edge
in \Fact d\ by the choice of the edges $e_i$ in Line~\ref{treeconstruction} of Algorithm~\ref{threeapx}.
\end{proof}

\begin{lemma} \label{lem:2ec}
Assume that $R$ is computed as in Algorithm~\ref{threeapx}.
Then $R$ is 2-edge-connected.
\end{lemma}

\begin{proof}
First, we observe that $R$ is connected: We do not remove any bridge edges and we remove at most one edge per 2-edge-connected component of $\Fact d$. Furthermore, the 2-edge-connected components that are not already connected in $\Fact d$ are connected via $H'$. 
To show that $R$ is 2-edge-connected, we show that the removal of a single edge does not disconnect the graph.

If we remove an edge $(u_i, v_{i+1})$, then we still have a connection between $u_i$ and $v_{i+1}$ via $u_{i+1}, u_{i+2}, \ldots, u_{i-2}, u_{i-1}$.
If we remove a bridge edge $e=(u,v)$ of \Fact d, then both $u$ and $v$ must each be connected to at least one 2-edge-connected component of $\Fact d$. Since those 2-edge-connected components are also connected through $H'$, $R$ remains to be connected:
First, if we remove a non-bridge edge $e$ within a 2-edge-connected component $C$ of \Fact d, then this also does not disconnect the graph. If $C$ does not correspond to a leaf in $T(\Fact d)$, then it is still 2-edge-connected in $R$.
Thus removing one edge does not disconnect the graph.
Second, if $C$ is a leaf in $T(\Fact d)$, then an edge $e=(u,v)$ is removed from $C$ in $R$. If removing another edge separates $C$ in two components, then $u$ and $v$ must be in separate components, but then these components are still connected through $H'$. Thus,
$R$ remains to be connected also in this case.
\end{proof}

\begin{lemma}
Assume that $R$ is computed as in Algorithm~\ref{threeapx}.
Then $w(R) \leq 3\cdot w(\RCS d) \leq 3\cdot w(\RTwoCS d)$.
\end{lemma}
\begin{proof}
The second inequality follows from Lemma~\ref{lem:undirected}\eqref{mstrcs}.
For the first inequality, we observe that $w(\MST) \leq w(\RCS d)$ (Lemma~\ref{lem:undirected}\eqref{mstrcs})
and $w(\Fact d) \leq w(\RCS d)$. Also, $w(\{u_i, v_{i+1}\}) \leq w(\{u_i, u_{i+1}\})
 + w(\{u_{i+1}, v_{i+1}\})$ by the triangle inequality.
We have
\begin{align*}
 w(R) & \leq \underbrace{w(H') + w(Q)}_{\text{adding the $\{u_i, v_{i+1}\}$'s}}
 - \underbrace{w(Q)}_{\text{removing $Q$}} + w(\Fact d) \\
& \leq w(H) + w(\Fact d) \leq 2\cdot w(\MST) + w(\Fact d) \leq 3\cdot w(\RCS d).
\end{align*}
The second-to-last inequality holds since we can obtain a TSP tour by duplicating all edges of an MST and taking shortcuts.
\end{proof}

The following theorem is an immediate consequence of the lemmas above.

\begin{theorem}
\label{thm:threeapx}
For all $d$, Algorithm~\ref{threeapx} is a polynomial-time 3-approximation for \MinRCS d\ and \MinRTwoCS d. This includes the case that $d$ is a function of $n$.
\end{theorem}

\begin{remark}
If we are only interested in a 3-approximation for \MinRCS d and not for \MinRTwoCS d, then we can simplify Algorithm~\ref{threeapx}
a bit: we only pick one non-bridge edge for each component and not for every 2-edge-connected component. The rest of the algorithm and its analysis remain the same. However, this does not seem to improve the worst-case approximation ratio.
\end{remark}

\begin{remark}
The analysis is tight in the following sense: By Lemma~\ref{lem:undirected}\eqref{onetwo}, a minimum-weight 2-edge-connected $d$-factor
can be three times as heavy as a minimum-weight connected $d$-factor. Thus, any algorithm that outputs a 2-edge-connected $d$-factor
cannot achieve an approximation ratio better than $3$.
Furthermore, since $w(\MST) \leq w(\RTwoCS d)$ and $w(\Fact d) \leq w(\RTwoCS d)$ are tight (Lemma~\ref{lem:undirected}\eqref{mstrcs} and \eqref{factrcs}),
the analysis is essentially tight. If we only require connectedness and not 2-edge-connectedness,
we see that the analysis cannot be improved since $w(\TSP) \leq 2 w(\RCS d)$ and $w(\Fact d) \leq w(\RCS d)$ are tight.

However, it is reasonable to assume that not all these inequalities can be tight at the same time and, in addition, shortcutting of the duplicated MST to obtain a TSP tour through $u_1,\ldots,u_k$
does not yield an improvement. Therefore, it might be possible to improve the analysis and show that Algorithm~\ref{threeapx} achieves a better approximation ratio than 3.
\end{remark}

\begin{remark}
Lines~\ref{dt1} and~\ref{dt2} of Algorithm~\ref{threeapx} are in fact simply the double-tree heuristic for \MinTSP~\cite[Section 2.4]{Shmoys2011}.
One might be tempted to construct a better tour using Christofides' algorithm~\cite[Section 2.4]{Shmoys2011}, which achieves
a ratio of $3/2$ instead of only $2$. However, in the analysis we compare the optimal solution for \MinRCS d to the MST,
and we know that $w(\MST) \leq w(\RCS d) \leq w(\RTwoCS d)$. If we use Christofides' algorithm directly, we have to compare
a TSP tour to the minimum-weight connected $d$-factor. In particular for odd $d$,
we have that for some instances $w(\TSP) \geq (\frac 43 - o(1)) \cdot w(\RTwoCS d) \geq (\frac 43 - o(1)) \cdot w(\RCS d)$ (Lemma~\ref{lem:undirected}\eqref{fourthird}).
Even if this is the true bound -- as it is for $d=3$ (Lemma~\ref{lem:undirected}\eqref{fourthird3}) --, the TSP tour constructed
contributes with a factor $3/2$ times $4/3$, which equals $2$, to the approximation ratio, which is no improvement.
\end{remark}

\subsection{$(r+1)$-Approximation}
\label{ssec:tspapprox}

In this section, we give an $(r+1)$-approximation for \MinRCS d for even values of $d$ and \MinARCS d for all values of $d$. Here, $r$
is the ratio within which $\MinTSP$ (for \MinRCS d) or \MinATSP\ (for \MinARCS d) can be approximated. This means that
we currently have $r = 3/2$ for the symmetric case by Christofides' algorithm~\cite[Section 2.4]{Shmoys2011}
and, for the asymmetric case, we have either $r = O(\log n/\log\log n)$ if we use the randomized
algorithm by Asadpour et al.~\cite{ATSPolog} or $r= \frac 23 \cdot \log_2 n$ if we use Feige and Singh's deterministic
algorithm~\cite{FeigeSingh}. Although the algorithm is a simple modification of Algorithm~\ref{threeapx}, we summarize it as Algorithm~\ref{rplus1}
for completeness.

\begin{algorithm}[t]
\SetKwInOut{Input}{input}\SetKwInOut{Output}{output}
\Input{undirected or directed complete graph $G=(V,E)$, edge weights $w$, $d$}
\Output{connected $d$-factor $R$ of $G$}

compute a minimum-weight $d$-factor $C$ of $G$\;

let $C_1, \ldots, C_k$ be the connected components of $C$, and let $e_i = (u_i, v_i)$ be any edge of $C_i$

compute a TSP tour $H$ using an approximation algorithm with ratio $r$\;

take shortcuts to obtain from $H$ a TSP tour $H'$ through $\{u_1, \ldots, u_k\}$, assume w.l.o.g.\
that $H'$ traverses the vertices in the order $u_1, \ldots, u_k, u_1$\;

obtain $R$ from $C$ by adding the edges $(u_i, v_{i+1})$ (with $k+1 = 1$) and removing $e_1, \ldots, e_k$\; \label{removeadd}
\caption{$(r+1)$-approximation for \MinRCS d\ for even $d$ and \MinARCS d.}
\label{rplus1}
\end{algorithm}

\begin{theorem}
If \MinTSP\ can be approximated in polynomial time within a factor of $r$, then Algorithm~\ref{rplus1} is a polynomial-time $(r+1)$-approximation
for \MinRCS d\ for all even~$d$.

If \MinATSP\ can be approximated in polynomial time within a factor of $r$, then Algorithm~\ref{rplus1} is a polynomial-time $(r+1)$-approximation
for \MinARCS d\ for all~$d$.

The results still hold if $d$ is part of the input.
\end{theorem}
\begin{proof}
Let $T$ be an optimal TSP tour, and let $O$ be an optimal connected $d$-factor. Let $C$ be a minimum-weight $d$-factor, as computed
by Algorithm~\ref{rplus1}.
We have $T = \ATSP$, $O = \ARCS d$, and $C = \AFact d$ if the input graph is asymmetric
and $T = \TSP$, $O = \RCS d$, and $C = \Fact d$ if the input graph is symmetric.

By Lemma~\ref{lem:undirected} and Lemma~\ref{lem:directed}, we have
$w(T) \leq w(O)$ and $w(C) \leq w(O)$. Removing and adding edges as in Line~\ref{removeadd} of Algorithm~\ref{rplus1}
yields again a $d$-factor. For the asymmetric case, any component is strongly connected. After removal of one edge per component,
it is still weakly connected. For the symmetric case, any component is 2-edge-connected.
Thus, the removal of edges in Line~\ref{removeadd} does not split any component. Hence, the addition of edges
in Line~\ref{removeadd} yields a connected $d$-factor $R$.
By the triangle inequality, we have $w(R) \leq w(C) + w(H') \leq w(C) + w(H)$.
Since we use an $r$-approximation to obtain $H$, we thus have
$w(R) \leq w(C) + r w(T) \leq (r+1) \cdot w(O)$.
\end{proof}

\section{Hardness Results}
\label{sec:hardness}

\subsection{\TSP-Inapproximability}
\label{ssec:tspinapprox}

In this section, we prove that \MinRCS d\ cannot be approximated better than \MinTSP.

\begin{theorem}
\label{thm:mininapx}
For every $d \geq 2$, if \MinRCS d\ can be approximated in polynomial time within a factor of $r$, then \MinTSP\ can be approximated in polynomial time
within a factor of $r$.
\end{theorem}

\begin{proof}
We show that $\MinRCS d$ can be used to approximate \MinTSP.
Let the instance of \MinTSP\ be given by a complete graph $G=(V,E)$ and edge weights $w = (w_e)_{e \in E}$ that satisfy the triangle inequality. Let $n = |V|$. We construct an instance of \MinRCS d as follows: The instance consists of a complete graph $H = (V', E')$. Here $V' = \bigcup_{v \in V} V_v$, where $V_v = \{v_1, v_2, \ldots, v_{d+1}\}$, i.e., $H$ contains $(d+1) \cdot n$ vertices.
We assign edge weights $\tilde w$ as follows:
\begin{itemize}
\item $\tilde w_{\{v_i, v_j\}} = 0$ for all $v \in V$, $i \neq j$,
\item $\tilde w_{\{u_i, v_j\}} = w_{\{u,v\}}$ for all $u \neq v$, $i$ and $j$.
\end{itemize}
Every TSP tour $T$ of $G$ maps to a connected $d$-factor $R$ of $H$ of the same weight:
We give $T$ an orientation. For an edge from $u$ to $v$ in $T$, we include $\{u_1,v_2\}$ in $R$.
Adding all edges except $\{v_1,v_2\}$ to $R$ within each $V_v$ yields a connected $d$-factor $R$. Clearly, $\tilde w(F) = w(T)$.

Now assume that we have a connected $d$-factor $R$ of $H$. We claim that we can construct a TSP tour $T$ of
$G$ with $w(T) \leq \tilde w(R)$.
We construct a multiset $T'$ of edges of $G$ as follows: For each edge $\{u_i,v_j\}$ of $R$, if $u\neq v$, we add an edge $\{u,v\}$ to $T'$. Otherwise, if $u=v$, we ignore the edge. The sum of the degrees in $R$ of all vertices in each set $V_v$ is equal to $(d+1)d$ and is therefore even. Thus, for each $v$, the number of edges leaving $V_v$ in $R$, which equals the number of edges incident to $v$ in $T'$ by construction, is even as well. Since $R$ is connected, the multigraph $G'=(V,T')$ is connected as well. By construction, $w(T') = \tilde w(R)$. Since $G'$ is connected and all its vertices have even degree, $G'$ is Eulerian. Therefore, we can obtain a TSP tour $T$ from $T'$ by taking shortcuts. By the triangle inequality, $w(T) \leq w(T') = \tilde w(R)$.
\end{proof}

The same construction as in the proof of Theorem~\ref{thm:mininapx} yields the same result
for \MinRTwoCS d. A similar construction yields the same result for \MinARCS d.

\begin{corollary}
\label{cor:mininapx2}
For every $d \geq 2$, if \MinRTwoCS d\ can be approximated in polynomial time within a factor of $r$, then \MinTSP\ can be approximated in polynomial time
within a factor of $r$.
\end{corollary}

\begin{corollary}
\label{cor:mininapx3}
For every $d \geq 2$, if \MinARCS d\ can be approximated in polynomial time within a factor of $r$, then \MinATSP\ can be approximated
in polynomial time within a factor of $r$.
\end{corollary}

\MinTSP, \MinATSP, \MaxTSP, and \MaxATSP\ are \APX-hard~\cite{PY93}. Furthermore, the
reduction from \MinTSP\ to \MinRCS d\ is in fact an L-reduction~\cite{PY91}
(see also Shmoys and Williamson~\cite[Section 16.2]{Shmoys2011}). This proves the \APX-hardness of 
\MinRCS d\ for all $d$.
The reductions from \MinTSP\ to \MinRTwoCS d\ and from \MinATSP\ to \MinARCS d work in the same way.
Furthermore, by reducing from \MaxTSP\ and \MaxATSP\ in a similar way (here, the edges between
the copies of a vertex have high weight), we obtain \APX-hardness for \MaxRCS d\ and \MaxARCS d\ as well.

\begin{corollary}
\label{cor:apxhard}
For every fixed $d \geq 2$,
the problems \MinRCS d, \MinRTwoCS d, and \MaxRCS d are \APX-complete.
For every fixed $d \geq 1$,
\MinARCS d and \MaxARCS d are \APX-complete.
\end{corollary}

\subsection{Hardness for Growing $d$}
\label{ssec:growingd}

In this section, we generalize the \NP-hardness proof for \RCS d\ by Cheah and Corneil~\cite{CC90}
to the case that $d$ grows with $n$.
Furthermore, we extend Theorem~\ref{thm:mininapx} and Corollaries~\ref{cor:mininapx2} and~\ref{cor:mininapx3} and
the \APX-hardness of the minimization variants (Corollary~\ref{cor:apxhard}) to growing $d$.
The \APX-hardness of \MaxRCS d\ and \MaxARCS d\ does not transfer to growing $d$ -- both can be approximated
within a factor of $1-O(1/d^2)$, which is $1-o(1)$ for growing $d$.

Let us consider Cheah and Corneil's~\cite[Section 3.2]{CC90} reduction from
$\RCS 2$, i.e., the Hamiltonian cycle problem,
to $\RCS d$. Crucial for their reduction is the notion of the $d$-expansion
of a vertex $v$, which is obtained as follows:
\begin{enumerate}
\item We construct a gadget $G_{d+1}$ by removing a matching of size $\lceil\frac{d}{2}\rceil-1$ from a complete graph on $d+1$ vertices.
\item We connect each vertex whose degree has been decreased by one to $v$.
\end{enumerate}

The reduction itself takes a graph $G$ for which we want to test
if $G \in \RCS 2$ and maps it to a graph $R_d(G)$ as follows:
For even $d$, $R_d(G)$ is the graph obtained by performing a $d$-expansion
for every vertex of $G$.
For odd $d$, the graph $R_d(G)$ is obtained by doing the following for each vertex $v$ of $G$:
add vertices $u_1,u_2,\ldots,u_{d-2}$; connect $v$ to $u_1, \ldots, u_{d-2}$;
perform a $d$-expansion on $u_1, \ldots, u_{d-2}$.
We have $G \in \RCS 2$ if and only if $R_d(G) \in \RCS d$.

We note that $R_d(G)$ has $(d+2)\cdot n$ vertices for even $d$ and $\Theta(d^2n)$ vertices for odd $d$
and can easily be constructed in polynomial time since $d<n$.

\begin{theorem}
\label{thm:growinghardness}
For every fixed $\eps > 0$, there is a function $f = \Theta(n^{1-\eps})$ that maps to even integers
such that $\RCS{f}$ is \NP-hard.

For every fixed $\eps > 0$, there is a function $f = \Theta(n^{\frac 12 - \eps})$ that maps to odd integers
such that $\RCS{f}$ is \NP-hard.
\end{theorem}

\begin{proof}
We first present the proof for the case that we map to even integers.
After that, we briefly point out the difference for odd integers.

We choose $d = 2 \lceil n^{\frac{1-\eps}\eps} \rceil$
and apply $R = R_d(G)$.
The graph $R$ has $g(n) = n \cdot (2 \lceil n^{\frac{1-\eps}\eps} \rceil +2)$
vertices since $d$ is even. We have $g = \Theta(n^{1/\eps})$. Now we determine $f$:
we require 
$f(g(n)) = d = 2 \lceil n^{\frac{1-\eps}\eps} \rceil$.
This can be achieved because $g = \omega(n)$ is an injective function.

Expressed as a function of $g$, we have
$d = \Theta(g(n)^{1-\eps})$. For natural numbers that are not images of $g$,
we interpolate $f$ to maintain the growth bound. Thus, $f(n) = \Theta(n^{1-\eps})$.

Let us now point out the differences for functions $f$ mapping to odd integers.
In this case, since the reduction for $d$ maps to graphs of size $\Theta(d^2 n)$, we have to choose
$d = \Theta(n^{\frac{1-\eps}{2\eps - 1}})$. This, however, works only up to $\eps > 1/2$
or functions up to $n^{\frac 12 - \eps}$.
\end{proof}

In the same way as the \NP-completeness, the inapproximability can be transferred.
The reduction creates graphs of size $(d+1) \cdot n$. The construction is the same
as in Section~\ref{ssec:tspinapprox}, and the proof follows the line
of the proof of Theorem~\ref{thm:growinghardness}. Here, however, we do not have to distinguish
between odd and even $d$ for the symmetric variant, as the reduction in Section~\ref{ssec:tspinapprox}
is the same for both cases.

\begin{theorem}
For every fixed $\eps > 0$, there is a function $f = \Theta(n^{1-\eps})$
such that $\MinRCS{f}$ and \MinRTwoCS{f}\ are \APX-hard and cannot be approximated better than \MinTSP.

For every fixed $\eps > 0$, there is a function $f = \Theta(n^{1-\eps})$
such that $\MinARCS{f}$ is \APX-hard and cannot be approximated better than \MinATSP.
\end{theorem}

\section{Further Algorithms}
\label{sec:special}
\subsection{2-Approximation for $d \geq n/3$}
\label{ssec:twoapprox}

If $d \geq n/3$, then we easily get a better approximation algorithm
for \MinRTwoCS d\ and \MinRCS d.
In this case, \Fact d\ consists either of a single component -- then we are done -- or of two components $C_1$ and $C_2$ with
$C_i = (V_i, E_i)$.
In the latter case, we proceed as follows: first, find the lightest edge $e = \{u, v\}$ with $u \in V_1$
and $v \in V_2$. Second, choose any edges $\{u, u'\} \in E_1$ and $\{v, v'\} \in E_2$.
Third, remove $\{u, u'\}$ and $\{v, v'\}$ and add $\{u,v\}$ and $\{u', v'\}$.
The increase in weight is at most $2 \cdot w(\{u,v\})$ by the triangle inequality.

The resulting graph is clearly $d$-regular. It is connected since $C_1$ and $C_2$
are 2-edge-connected: they both consist of at most $\frac{2n}3 -1$ vertices and
are $d$-regular with $d \geq n/3$. Thus, they are even Hamiltonian by
Dirac's theorem \cite{DW00}.
Furthermore, any connected $d$-regular graph must have at least two edges connecting
$V_1$ and $V_2$: If $d$ is even, then this follows by 2-edge-connectedness.
If $d$ is odd, then $|V_1|$ and $|V_2|$ are even and, thus, an even number of edges
must leave either of them.
Thus, $w(\{u,v\}) \leq \frac 12 \cdot w(\RCS d)$.
Since we add at most $2 \cdot w(\{u,v\})$ and also have $w(\Fact d) \leq w(\RCS d)$,
we obtain the following theorem.

\begin{theorem}
For $d \geq n/3$, there is a polynomial-time
2-approximation for \MinRCS d.
\end{theorem}

\subsection{Decision Problem for $d = \lceil \frac n2 \rceil -1$}
\label{ssec:decision}

For $d \geq n/2$, any $d$-factor is immediately connected and also the minimization variant can be solved efficiently.
In this section, we slightly extend this to the case of $d \geq \frac n2 -1$.

We assume that the input graph $G$ is connected. To show that the case $d=\lceil \frac n2 \rceil -1$ is in \DP, we compute a $d$-factor. If none exists or we obtain a connected $d$-factor,
then we are done. Otherwise, we have a $d$-factor consisting of two components $C_1$ and $C_2$ which are both cliques of size $n/2$.
If $G$ contains a cut vertex, say, $u \in C_1$, then this is the only vertex with neighbors in $C_2$. In this case, $G$ does not contain a connected $d$-factor. If $G$ does not contain a cut vertex, there are
two disjoint edges $e=\{u,v\}$, $e'=\{u', v'\}$
with $u, u' \in C_1$ and $v, v' \in C_2$. Adding $e$ and $e'$ and removing $\{u, u'\}$
and $\{v, v'\}$ yields a connected $d$-factor.

\begin{theorem}
\RCS d\ is in $\DP$ for every $d$ with $d \geq \frac n2 -1$.
\end{theorem}

\subsection{Approximating \MaxARCS d}
\label{ssec:directedapx}

The approximation algorithm for \MaxRCS d\ \cite{BG:2004} can easily be adapted to work for
\MaxARCS d: We compute a directed $d$-factor of maximum weight.
Any component consists of at least $d+1$ vertices, thus at least $d \cdot (d+1)$ arcs.
We remove the lightest arc of every component and connect the resulting (still at least weakly connected)
components arbitrarily to obtain a connected $d$-factor.
Since we have removed at most a $\frac 1{d\cdot (d+1)}$-fraction of the weight, we obtain the following result.

\begin{theorem}
For every $d$, \MaxARCS d\ can be approximated within a factor of $1-\frac{1}{d \cdot (d+1)}$.
\end{theorem}

\section{Open Problems}
\label{sec:open}

An obvious open problem is to improve the approximation
ratios. Apart from this, let us mention two open problems: First, is it possible to achieve constant factor
approximations for minimum-weight $k$-edge-connected or $k$-vertex-connected $d$-regular graphs?
Without the regularity requirement, the problem of computing minimum-weight $k$-edge-connected graphs can be approximated
within a factor of $2$~\cite{KhullerV94} and the problem of computing minimum-weight $k$-vertex-connected graphs can
be approximated within a factor of $2 + 2 \cdot \frac{k-1}n$ for metric instances~\cite{KR1996}
and still within a factor of $O(\log k)$ if the instances are not required to satisfy the triangle inequality~\cite{kVCSS}.
We refer to Khuller and Raghava\-cha\-ri~\cite{KhullerR08} for a concise survey.

Second, we have seen that $\RCS{(\lceil \frac n2 \rceil -1)} \in \DP$, but we do not know if \MinRCS{(\lceil \frac n2 \rceil -1)} can be solved in polynomial time as well.
In addition, we conjecture that also $\RCS{(\lceil \frac n2 \rceil - k)}$ is in \DP\ for any constant $k$.


\end{document}